\documentclass[11pt]{article}
\usepackage{fullpage,amsmath,amsthm,amssymb,bm,hyperref,cleveref,graphicx,url}
\usepackage[linesnumbered,ruled,vlined]{algorithm2e}
\usepackage{authblk} 
\usepackage{subcaption} 
\usepackage{rotating}

\graphicspath{{./Fig/}}

\newtheorem{theorem}{Theorem}

\def\Poi{\mathrm{Poi}}
\def\Nor{\mathrm{Nor}}
\def\Cat{\mathrm{Cat}}
\def\vMF{\mathrm{vmF}}

\def\CDF{\mathrm{CDF}}
\def\ess{\mathrm{ess}}
\def\JS{\mathrm{JS}}
\def\dtheta{\mathrm{d}\theta}
\def\dmu{\mathrm{d}\mu}
\def\eps{\epsilon}
\def\dP{\mathrm{d}P}
\def\dQ{\mathrm{d}Q}
\def\sign{\mathrm{sign}}

\def\dx{\mathrm{d}x}

\def\bbR{\mathbb{R}}
\def\KL{\mathrm{KL}}
\def\TV{\mathrm{TV}}

\def\calX{\mathcal{X}}

\def\calD{\mathcal{D}}

\def\calF{\mathcal{F}}
\def\eqdef{{:=}}
\def\bbN{\mathbb{N}}

\def\st{{\ :\ }}

\def\impformula#1{\boxed{#1}}

\title{On power chi expansions of $f$-divergences}

\author[$\star$]{Frank Nielsen}
\affil[$\star$]{Sony Computer Science Laboratories, Inc.}
\affil[$\star$]{Tokyo, Japan}
\affil[$\star$]{{\small\tt Frank.Nielsen@acm.org}}
\author[$\dagger$]{Ga\"etan Hadjeres}
\affil[$\dagger$]{Sony Computer Science Laboratories, Inc.}
\affil[$\dagger$]{Paris, France}
\affil[$\dagger$]{{\small\tt Gaetan.Hadjeres@sony.com}}

\date{}

\begin{document}
\maketitle

\begin{abstract}
We consider both finite and infinite power chi expansions of $f$-divergences derived from Taylor's expansions of smooth generators, and
elaborate on cases where these expansions yield closed-form formula, bounded approximations, 
or analytic divergence series expressions of $f$-divergences.
\end{abstract}

 
\section{Introduction}

\subsection{Statistical $f$-divergences}

Let $(\calX,\calF)$ be a measurable space~\cite{measure-2013}, where $\calX$ denotes the sample space and $\calF$ 
is a $\sigma$-algebra of measurable events on $\calX$.
For a convex function $f:(0,\infty)\rightarrow \bbR$, strictly convex at $1$, the {\em $f$-divergence}~\cite{Csiszar-1967,Liese-2006} between two probability measure $P$ and $Q$ is defined  by:
\begin{equation}\label{eq:ifm}
I_f(P:Q) \eqdef \left\{
\begin{array}{ll}
\int_\calX  f\left(\frac{\dQ}{\dP}\right) \dP,  & Q\ll P\\
+\infty & Q\not\ll P.
\end{array}
\right.,
\end{equation}
where $Q\ll P$ denoting that $Q$ is absolutely continuous~\cite{measure-2013} with respect to $P$.
Notice that the integral defining the $f$-divergence may potentially diverge: In that case, we have $I_f(P:Q)=+\infty$.
The strict convexity  of the generator $f$ at $1$ is required for satisfying the law of the indiscernibles of the $f$-divergence: $I_f(P:Q)=0$ if and only if $P=Q$ almost everywhere (a.e.).
Since it follows from Jensen's inequality that $I_f(P:Q)\geq f(1)$, we assume in the remainder that $f(1)=0$.
Moreover, let $f_\alpha(u)=f(u)+\alpha(u-1)$. We have $I_{f_\alpha}=I_f$ for any $\alpha\in\bbR$.
Thus we fix the generator $f_\alpha$ that satisfies $f'_\alpha(1)=0$.
In information geometry, a {\em standard} $f$-divergence~\cite{IG-2016} further satisfies the scaling normalization $f''_\alpha(1)=1$ (see Appendix~\ref{sec:stdfdiv} for details).

When the probability measures $P$ and $Q$ are defined on the same measure space $(\calX,\calF,\mu)$ where $\mu$ is a base $\sigma$-finite positive measure
(with $P,Q\ll \mu$), let $p$ and $q$ denote their respective Radon-Nikodym densities~\cite{measure-2013}: 
$p=\frac{\dP}{\dmu}$ and $q=\frac{\dQ}{\dmu}$.
Then the $f$-divergence between $p$ and $q$ is then defined by
\begin{equation}\label{eq:if}
I_f^\mu(p:q) \eqdef \int_\calX p(x) f\left(\frac{q(x)}{p(x)}\right) \dmu(x) \geq f(1)=0.
\end{equation}
We adopt the following conventions for the mathematically undefined expressions in Eq.~\ref{eq:if}:
$$
f(0)=\lim_{u\rightarrow 0^+} f(u),\quad 0 f\left(\frac{0}{0}\right)=0,\quad
 0f\left(\frac{a}{0}\right)=\lim_{u\rightarrow 0^+} \left(\frac{a}{u}\right).
$$
It can be shown that the $f$-divergence between $P$ and $Q$  is {\em independent} of the choice of the dominating measure $\mu$ (in particular, one can choose $\mu=\frac{P+Q}{2}$. Thus we write concisely $I_f(p:q)$ for $I_f^\mu(p:q)$. 
In the remainder, we consider all measures dominated by the base measure $\mu$ (e.g., the Lebesgue or counting measures), and
function $f$ shall be called the {\em generator} of the divergence (with $f(1)=0$ and $f'(1)=0$).
 
Let us give some examples of $f$-divergences in disguise:
The {\em total variation} distance $\TV(p,q)=\frac{1}{2}\int |p(x)-q(x)| \dmu(x)$ is a metric $f$-divergence obtained for the generator $f_\TV(u)=|u-1|$ which is strictly convex at $u=1$ (and only convex otherwise).
Another example is the {\em chi-squared} distance~\cite{chisquared-1900} which is $f$-divergences obtained for the generator $f_{\chi^2}(u)=(u-1)^2$:
$$
\chi^2(p:q) =  \int \frac{(q(x)-p(x))^2}{p(x)} \dmu(x).
$$

In theory, the definite integral of the $f$-divergences of Eq.~\ref{eq:if} may be computed using Risch semi-algorithm of symbolic integration~\cite{SI-2006}.
However, this semi-algorithm  requires to test whether some expressions are equivalent to zero or not. 
It is not known whether such an algorithm exists or not. 
Worse, when the absolute value function belongs to the set of the elementary functions,
it can be proved that no such algorithm exists.\footnote{\url{https://en.wikipedia.org/wiki/Risch_algorithm}}
Thus in practice, one often relies on stochastic Monte Carlo integrations for estimating $f$-divergences, or on various approximation bounds~\cite{GMMLSE-2016}.

\subsection{Power chi pseudo-distances: Closed-form formula for affine exponential families}

Let us define the following {\em $i$-th power chi pseudo-distance}:
\begin{equation}
\chi_i^\pm(p:q) = \int \frac{(q(x)-p(x))^i}{p(x)^{i-1}} \dmu(x), \quad i\in\{2, 3,\dots\}.
\end{equation}
We have $\chi_2^\pm(p:q)=\chi^2(p:q)$.
Although $\chi_i^\pm$ is a distance for even integer $i\geq 2$,
it is {\em not} a distance for odd $i$ as $\chi_i^\pm$ maybe negative (hence the notation $\chi_i^\pm$) and even fail to satisfy the law of the indiscernibles (i.e., $\chi_i^\pm(p:q)=0$ if and only if $p=q$).
To see this, consider binary categorical distributions $p=(\lambda_p,1-\lambda_p)$ and $q=(\lambda_q,1-\lambda_q)$.
Then we have
$$
\chi_i^\pm(p:q) = \frac{(\lambda_q-\lambda_p)^i}{\lambda_p^{i-1}} + \frac{(1-\lambda_q-1+\lambda_p)^i}{(1-\lambda_p)^{i-1}}.
$$
When $i$ is odd, we get
$$
\chi_i^\pm(p:q) = (\lambda_q-\lambda_p)^i  \left( \frac{1}{\lambda_p^{i-1}} -    \frac{1}{(1-\lambda_p)^{i-1}})  \right).
$$
Thus for $\lambda_p=1-\lambda_p=\frac{1}{2}$, we have $\chi_i^\pm(p:q)=0$ for odd $i$, {\em independent} of the values of $\lambda_q$.

In the remainder, we term the $\chi_i^\pm$'s the $i$-order chi {\em pseudo-distances} or the {\em power chi distances} for short.

Notice that the $\chi_i^\pm$'s are $f$-divergences for even integer $i$ for the generators $f_{\chi_i^\pm}(u)=(u-1)^{i}$.

Let us further extend the definition as follows (for any $\lambda\not=0$):  
\begin{equation}
\chi_{i,\lambda}^\pm(p:q) \eqdef \int \frac{(q(x)-\lambda p(x))^i}{p(x)^{i-1}} \dmu(x), \quad i\in\{2, 3,\dots\}.
\end{equation}
We note that $\chi_{i,1}^\pm(p:q)=\chi_{i}^\pm(p:q)$.
We shall define $\lambda^j=\sign(\lambda)^j |\lambda|^j$ for any integer $j\in\bbN$, where 
$$
\sign(\lambda)=\left\{
\begin{array}{ll}
1 & \mbox{if $\lambda>0$},\\
-1 & \mbox{if $\lambda<0$},\\
0 & \mbox{if $\lambda=0$}. 
\end{array}
\right.
$$

An full natural exponential family $\{f(x;\theta)\dmu\}$ is a set of probability distributions with densities with respect to a base measure $\mu$
written canonically~\cite{EF-2009,ECE-EF-2010} as $f(x;\theta)=\exp(t(x)^\top\theta-F(\theta)+k(x)) 1_{\calX}(x)$,
where $\calX$ is the support, $t(x)$ a vector of sufficient statistics, $\theta$ the natural parameter, $k(x)$ an auxiliary carrier measure term and $F(\theta)$ the log-normalizer (also called log-partition function or cumulant function). 
The order $D$ of the family is the dimension of $t(x)$ or $\theta$.
The natural parameter space~\cite{EF-2009} is $\Theta = \{\theta \st \int f(x;\theta)\dmu<\infty\}$ (with $\Theta\subset \bbR^D$).
Following~\cite{ChiFDiv-2014}, we report a closed-form formula for exponential families~\cite{EF-2009} 
with a natural parameter space being {\em affine}. 
These exponential families are termed {\em Affine Exponential Families}, or AEF for short.

Affine exponential families~\cite{EF-2009} include the isotropic Gaussian family, the Poisson family, the multinomial family, the von Mises-Fisher distributions\footnote{See~\cite{EFvonMises-2013} for the canonical parameterization of this directional statistics exponential family  which differs from the usual polar coordinate parameterization.}, etc.
See Table~\ref{tab:aef}.

\begin{table}
\centering

\begin{eqnarray*}
\Poi(\lambda) &:& f(x;\lambda)=\frac{\lambda^x e^{-\lambda}}{x!}, \lambda>0, x\in\{0, 1, ...\}\\
\Cat(p) &:& f(x;p) = p_0^{x_1}\ldots p_{d}^{x_{d}}, x_i\in\{0,1\}, \sum_{i=0}^{d} x_i=1  \\ 
\vMF(\theta) &:& f(x;\theta) = \frac{\exp(x^\top\theta)}{{}_0F_1(;\frac{d}{2};\frac{\|\theta\|^2}{4})} \\
\Nor(m,I) &:& f(x;\mu)= (2\pi)^{-\frac{d}{2}} e^{-\frac{1}{2}(x-m)^\top (x-m)}, m\in\mathbb{R}^d, x\in\mathbb{R}^d\\
\end{eqnarray*}

\begin{tabular}{lllll}
Density & $\theta$ & Domain & Log-normalizer $F(\theta)$ & Base measure\\ \hline\hline
Categorical & $\theta_i=\log\frac{p_i}{p_0}, i\in\{1,\ldots,d\}$ & $\mathbb{R}^d$  & $\log(1+\sum_{i=1}^d e^{\theta_i})$ & $\mu_c$\\
Poisson & $\log\lambda$ & $\mathbb{R}$ & $e^\theta$ &   $\mu_c$\\
Isotropic Gaussian & $m$ & $\mathbb{R}^d$ & $\frac{1}{2}\theta^\top\theta$ & $\mu_L(d)$\\ 
von-Mises Fisher & $\theta$ & $\mathbb{R}^{d-1}$ & $\log {}_0F_1(;\frac{d}{2};\frac{\|\theta\|^2}{4})$ & $\mu_S(d-1)$\\
\hline
\end{tabular}

\caption{Examples of exponential families with affine natural parameter spaces.
$\mu_c$ denotes the counting measure, $\mu_L(d)$ the Lebesgue measure on $\mathbb{R}^d$, and $\mu_S(d-1)$ the uniform measure on the $(d-1)$-dimensional sphere.
$_0F_1$ is a generalized hypergeometric series~\protect\cite{EFvonMises-2013}.
\label{tab:aef}}
\end{table}

Let $p=f(x;\theta_p)$ and $q=f(x;\theta_q)$ be two densities belonging to the same affine exponential family.
Then we have:
\begin{equation}\label{eq:cfchipm}
\impformula{
\chi^\pm_{i,\lambda}(p:q) = 
 \sum_{j=0}^i   (-\lambda)^{i-s} \binom{i}{j}   \exp\left({F\left((1-j)\theta_p+j\theta_q\right)-\left((1-j)F(\theta_p)+jF(\theta_q)\right)}\right).
}
\end{equation}
This can be easily seen be writing $\chi_i^\pm$  as $\chi_i^\pm(p:q) = \int p(x)^{1-i}(q(x)-p(x))^i  \dmu(x)$ and plugging the canonical densities of exponential families.
The binomial coefficients can be computed efficiently using Pascal's triangle.

Furthermore, this formula can be extended to $\chi^\pm_i(p:\sum_{i=1}^{l} w_i q_i)$ where $\sum_{i=1}^{l} w_iq_i$ is a mixture of exponential families by using the multinomial expansion~\cite{StatMinkowski-2019} instead of the binomial expansion:

\begin{eqnarray*}
\chi_{i,\lambda}^\pm\left(p:\sum_{i=1}^{l} w_i q_i\right) &=& \int \frac{(\sum_{i=1}^{l} w_i q_i(x)-\lambda p(x))^i}{p(x)^{i-1}} \dmu(x),\\
 &=& \int  \left( \sum_{s=0}^i \binom{i}{s} (-\lambda p(x))^{1-s} (\sum_{u=1}^{l} w_u q_u(x))^{s}  \right) \dmu(x),\\
 &=&  \sum_{s=0}^i \binom{i}{s}  (-\lambda)^{i-s}
 \sum_{\substack{\sum_{i=1}^l \alpha_i=s\\ \alpha_i\in\bbN}} 
\binom{\alpha}{\alpha_1,\ldots,\alpha_l} \int p(x)^{1-s} \prod_{u=1}^l (w_u q_u(x))^{\alpha_j} \dmu(x).
\end{eqnarray*}

Following~\cite{StatMinkowski-2019}, when $p(x)=f(x;\theta)$ and $q_u(x)=f(x;\theta_u)$ are densities of the same affine exponential family, we have
\begin{eqnarray}
\lefteqn{\int p(x)^{1-s} \prod_{u=1}^l (w_u q_u(x))^{\alpha_j} \dmu(x)=}\nonumber\\
&& \exp\left(F\left((1-s)\theta+\sum_{u=1}^l \alpha_u\theta_u\right)
- \left((1-s)F(\theta)+\sum_{u=1}^l \alpha_u F(\theta_u)\right)
 \right).
\end{eqnarray}

\begin{theorem}
The power chi pseudo-distances between a member and a mixture of an affine exponential family can be calculated in closed-form.
\end{theorem}

For isotropic Gaussian distributions $p\sim N(m_1,I)$ and $q\sim N(m_2,I)$ (where $I$ denotes the identity matrix), we get:
\begin{equation} \label{eq:hochigaussian}
\chi_i^\pm(p:q) = \chi_i^\pm(m_1:m_2) = \sum_{j=0}^i (-1)^{i-j} \binom{i}{j}   \exp\left(\frac{j(j-1)}{2}\|m_1-m_2\|^2\right).
\end{equation}
 
For any prescribed integer $i$, $\chi_i^\pm(m_1:m_2)<\infty$ and we can bound  $\chi_i^\pm$ as follows:
$$
|\chi_i^\pm(m_1:m_2)| \leq 2^i \exp\left(\frac{i(i-1)}{2}\|m_1-m_2\|^2\right) <\infty.
$$

Consider $\|m_1-m_2\|=1$ (e.g., in 1D, $m_1=0$ and $m_2=1$), then we have
\begin{eqnarray*}
\chi_2^\pm(0:1)&\simeq& 1.718281828459045,\\
\chi_3^\pm(0:1)&\simeq& 13.930691437810532,\\
\chi_4^\pm(0:1)&\simeq& 336.3963367707387,\\
\chi_5^\pm(0:1)&\simeq& 20186.99437829033,\\
\chi_6^\pm(0:1)&\simeq& 3142544.0730946246,\\
\chi_7^\pm(0:1)&\simeq& 1.2963817005597024\ 10^{9},\\
\chi_8^\pm(0:1)&\simeq& 1.4357968646042915\ 10^{12},\\
\chi_9^\pm(0:1)&\simeq& 4.298262439031654\ 10^{15},\\
\chi_{10}^\pm(0:1)&\simeq& 3.489122366600497\ 10^{19},\\
\mbox{etc.}
\end{eqnarray*}

We can notice that the  $\chi$ pseudo-distances diverge very quickly with the increase of the order.

\section{Power chi expansions of $f$-divergences from Taylor's theorem}

\subsection{Power chi expansions}

Consider the Taylor's $k$-th order expansions~\cite{FdivTaylor-2002} of the smooth generator $f$ around point $c$:
$$
f(u)= P_k(u) +  R_{k}(u),
$$
where $P_k(u)$ is the $k$-th order {\em Taylor polynomial}:
$$
P_k(u) = \sum_{i=0}^k \frac{f^{(i)}(c)}{i!} (u-c)^i,
$$
and
$R_{k}(u)=f(u)-P_k(u)$ is the {\em Taylor remainder} term.
There are many ways to express (e.g., Peano, Lagrange, Cauchy, or integral forms) and bound the remainder $R_k$ in a Taylor expansion~\cite{FdivTaylor-2002,FdivTaylor-2005}.

In particular, when $|f^{(k+1)}(u)|\leq M$ for any $u\in (c-r,c+r)$ (with $r>0$), then we have
$$
|R_k(u)|\leq M \frac{r^{k+1}}{(k+1)!}.
$$ 

Letting $c=1$ and $u=\frac{q(x)}{p(x)}$, and using the fact that $f(1)=f'(1)=0$, we get
$$
p(x) f\left(\frac{q(x)}{p(x)}\right) =  \sum_{i=2}^k \frac{f^{(i)}(1)}{i!} \frac{(q(x)-p(x))^i}{p(x)^{i-1}} 
+ R_k\left(\frac{q(x)}{p(x)}\right).
$$

Carrying the integral over the support $\calX$, we end up with the following {\em power chi expansion} of the $f$-divergence:
$$
I_f(p:q) = \sum_{i=2}^k \frac{f^{(i)}(1)}{i!} \chi_i^\pm(p:q) + \int_\calX  R_k\left(\frac{q(x)}{p(x)}\right) \dmu(x).
$$ 

Define the {\em $k$-th order chi expansion} of the $f$-divergence:
\begin{equation}
I_{f,k}^\chi(p:q) \eqdef \sum_{i=2}^k \frac{f^{(i)}(1)}{i!} \chi_i^\pm(p:q).
\end{equation}
It is ``chinomial'' of order $k$.
Let us express the power chi remainder as:
$$
R_{k}^\chi(p:q)=I_f(p:q) - I_{f,k}^\chi(p:q).
$$

Then we can approximate the $f$-divergence by its $k$-th order chi expansion, and bound the error as:
\begin{equation}
\impformula{
\left|I_f(p:q) - I_{f,k}^\chi(p:q) \right| = \left| R_{k}^\chi(p:q) \right| \leq \left| \int_\calX  R_k\left(\frac{q(x)}{p(x)}\right) \dmu(x)  \right|.
}
\end{equation}

\subsection{Some illustrating examples of power chi expansions}

Let us report the $k$-th order chi expansions for the  Kullback-Leibler divergence, the Jeffreys divergence and the Jensen-Shannon divergence:

\begin{itemize}
	\item The Kullback-Leibler (KL) divergence\footnote{In the literature, a divergence is either a statistical distance or a smooth parameter distance used to define an information-geometric manifold~\cite{IGDiv-2010}.} between two densities $p$ and $q$ is defined by
$$
\KL(p:q)=\int p(x)\log\frac{p(x)}{q(x)} \dmu(x).
$$
The KL divergence is a $f$-divergence obtained for the generator $f_\KL(u)=-\log u$.
It is an unbounded divergence 
that may potentially diverge even when distributions are defined on the same support.\footnote{
Let $\calX=(0,1)$ and two densities (with respect to Lebesgue measure $\dx$) $p_1(x)=1$ and $p_2(x)=c e^{-1/x}$ with $c^{-1}=\int_0^1 e^{-1/x}\dx \simeq 0.148$ the normalizing constant. Then
$\KL(p_1:p_2) =  \int_0^1 x_1 \log \frac{p_1(x)}{p_2(x)} \dx =-\log c + \int_0^1 \frac{1}{x}\dx = \infty$.
}
We have 
$$
f^{(i)}_\KL(u) ={(-1)}^{i} (i-1)! u^{-i},\quad i\geq 2,
$$ 
and it follows the $k$-order power chi expansion of the KL divergence:

\begin{equation}
\impformula{
\KL_k^\chi(p:q) = \sum_{i=2}^k \frac{(-1)^i}{i} \chi_i^\pm(p:q).
}
\end{equation}

The {\em reverse $f$-divergence} $I_f^r(p:q)=I_f(q:p)$ is obtained for the {\em conjugate} generator $f^r(u)=uf\left(\frac{1}{u}\right)$.
The {\em reverse KL divergence} is a $f$-divergence for the generator
$f_{\KL^r}(u)=u\log u$.
The derivatives of the generator are:
$$
f^{(i)}_{\KL^r}(u) ={(-1)}^{i} (i-2)! u^{-(i-1)},\quad i\geq 2.
$$ 

\begin{equation}
\impformula{
{\KL^r}_k^\chi(p:q) = \sum_{i=2}^k  {(-1)}^{i} \frac{1}{i(i-1)} \chi_i^\pm(p:q).
}
\end{equation}

In general, the derivatives of the conjugate generator are
$$
{f^r}^{(i)}(u) = u \left( \frac{d^i}{\mathrm{d}u^i} f\left(\frac{1}{u}\right)\right) + i\left(\frac{d^{i-1}}{\mathrm{d}u^{i-1}} f\left(\frac{1}{u}\right) \right), \quad  i\geq 1. 
$$
We get a closed-form expression by applying the formula of Fa\`{a} di Bruno~\cite{diBruno-2005}:

\begin{eqnarray*}
\lefteqn{{\frac{d^n}{dx^n}} f(g(x))=}\\
&&\sum_{1\cdot m_1+2\cdot m_2+3\cdot m_3+\cdots+n\cdot m_n=n} 
\frac{n! \  f^{(m_1+\cdots+m_n)}(g(x))  \prod_{j=1}^n\left(g^{(j)}(x)\right)^{m_j}}{m_1!\,1!^{m_1}\,m_2!\,2!^{m_2}\,\cdots\,m_n!\,n!^{m_n}}.
\end{eqnarray*}

\item The Jeffreys divergence is the following symmetrization of the KL divergence:
$$
J(p:q) = \KL(p:q)+ \KL(q:p) = \int \left(p(x)-q(x)\right)\log \frac{p(x)}{q(x)} \dmu(x).
$$
It is a $f$-divergence obtained for the generator  $f_J(u)=(u-1)\log u$ (with $f_J(1)=0$, $f_J'(u)=1-\frac{1}{u}+\log u$ and $f_J'(1)=0$), and the higher-order derivatives of the generator are
$$
f_J(u)^{(i)} = (-1)^{i} (i-2)! (u+i-1) \frac{1}{u^i}, \quad i\geq 2.
$$
Thus the $k$-order power chi expansion for the Jeffreys divergence is:
$$
\impformula{
J_k^\chi(p:q) = \sum_{i=2}^k (-1)^{i}  \frac{1}{i-1}  \chi_i^\pm(p:q).
}
$$

\item The Jensen-Shannon (JS) divergence~\cite{JS-1991} is defined by:
\begin{eqnarray}
\JS(p,q) &=& \frac{1}{2} \left( \KL\left(p:\frac{p+q}{2}\right) + \KL\left(q:\frac{p+q}{2}\right)\right)\\
&=& \frac{1}{2}\int \left(p(x)\log \frac{2p(x)}{p(x)+q(x)} +  q(x)\log \frac{2q(x)}{p(x)+q(x)}\right)\dmu(x),\\
&=& h\left(\frac{p+q}{2}\right)-\frac{h(p) +h(q) }{2},\label{eq:jsh}
\end{eqnarray}
where $h(p)=-\int p(x)\log p(x)\dmu(x)$ denotes Shannon differential entropy.
It is a bounded symmetrization\footnote{Indeed,  we have $p(x)\log \frac{2p(x)}{p(x)+q(x)}\leq p(x)\log 2$, and hence $\KL\left(p:\frac{p+q}{2}\right)\leq\log 2$.} of the KL divergence (i.e., $\JS(p,q)\leq \log 2$) obtained for the generator 
$$
f_\JS(u)=-(u+1)\log \frac{1+u}{2} + u\log u.
$$
The JS divergence gained interest in Deep Learning (DL), notably with the  Generative Adversarial Network (GAN) architecture~\cite{GAN-2014,fgan-2016}.
We have 
$$
f^{(i)}_\JS(u) = (-1)^{i-2} (i-2)! \left(\frac{1}{u^{i-1}}-\frac{1}{(u+1)^{i-1}}\right),\quad i\geq 2,
$$ 
and 
it follows the $k$-order power chi expansion of the JS divergence:
\begin{equation}
\impformula{
\JS_k^\chi(p:q) = \sum_{i=2}^k (-1)^{i-2} \frac{1}{i(i-1)} \left(1-\frac{1}{2^{i-1}}\right)  \chi_i^\pm(p:q).
}
\end{equation}

\end{itemize}

Note that the power chi expansions (e.g., $\JS_k^\chi$) are not distances since they do not satisfy the law of the indiscernibles.


There are many other $f$-divergences, some of them are not standard.\footnote{A {\em standard $f$-divergence}~\cite{IG-2016} satisfies $f(1)=f'(1)=0$ with $f''(1)=1$.}
For example, the harmonic divergence $H(p:q)=\int \frac{2p(x)q(x)}{p(x)+q(x)}\dmu(x)$ is an uncalibrated $f$-divergence 
 obtained for the generator $f_H(u)=\frac{2u}{u+1}$. Indeed, in that case, we have $f_H(1)=1$ and $f_H'(1)=\frac{1}{2}$.
The harmonic divergence is lower bounded by $f_H(1)=1$ and upper bounded by $2$. 
The higher-order derivatives are $f_H^{(i)}(u)=2(-1)^{i+1} i! \frac{1}{(u+1)^{i+1}}$.
The power chi expansion yields
\begin{equation}
H_k^\chi(p:q) =  1 + \sum_{i=2}^k (-1)^{i+1} \frac{1}{2^{i}}   \chi_i^\pm(p:q).
\end{equation}

Finally, the $f$-divergence defined for the generator 
$f_{\mathrm{exp}}(u)=e^u-eu$ (with $f_{\mathrm{exp}}(1)=0$ and $f_{\mathrm{exp}}'(1)=0$) has very simple higher-order derivatives: 
$f_{\mathrm{exp}}^{(i)}(u)=\exp(u)$ for $i\geq 2$.
Therefore the $k$-th power chi expansion is
$$
I_{f_{\mathrm{exp}}}^\chi(p:q)  = E_k^\chi(p:q)= \sum_{i=2}^k  \frac{e}{i!}   \chi_i^\pm(p:q).
$$
We call this $f$-divergence the exponential divergence for short.

\subsection{Finite  chi expansions for $f$-divergences with polynomial generators}

\begin{theorem}[Polynomial generator]\label{thm:polygen}
The $f$-divergence for any polynomial generator $f(u)=\sum_{j=0}^d a_ju^j$ has a finite chi expansion.
When the distributions are members of an affine exponential family, the $f$-divergence is calculated in closed-form.
\end{theorem}

\begin{proof}
Let $f(u)=\sum_{j=0}^d a_ju^j$ be a convex polynomial generator for the $f$-divergence of order $d$ (with coefficients satisfying $f(1)=f'(1)=0$).
We have $f^{(i)}(u)=0$ when $i>d$, and $f^{(i)}(u)= \sum_{j=i}^d  a_j \frac{j!}{(j-i)!} u^{j-i}$ for $0\leq i \leq d$.
It follows that $I_f$ has a finite chi expansion:
\begin{equation}
\impformula{
I_f(p:q)= \sum_{i=2}^d \left(\sum_{j=i}^d  a_j  \binom{j}{i}\right) \chi^\pm_i(p:q).
}
\end{equation}
\end{proof}

The reverse $f$-divergence $I_f^r(p:q)=I_f(q:p)$ is obtained for the conjugate generator $f^r(u)=uf\left(\frac{1}{u}\right)$.
The conjugate polynomial generator is $f(u)=\sum_{j=0}^d a_ju^{1-j}=\sum_{j=0}^d a_j \frac{1}{u^{j-1}}$ (called a Laurent polynomial).

A necessary condition for a polynomial to be convex on $\bbR_+$ is to have the leading coefficient positive.

\subsection{The case of $\alpha$-divergences}

The $\alpha$-divergences for $\alpha\in\bbR\backslash\{\pm 1\}$ are defined by:
$$
I_\alpha(p:q)=\frac{4}{1-\alpha^2} \left(1-\int p^{\frac{1-\alpha}{2}}(x) q^{\frac{1+\alpha}{2}}(x) \dmu(x)\right)=I_{-\alpha}(q:p).
$$
We have $\lim_{\alpha\rightarrow -1} I_\alpha(p:q)=\KL(p:q)$ and $\lim_{\alpha\rightarrow 1} I_\alpha(p:q)=\KL^r(p:q)=\KL(q:p)$ (reverse KL divergence).

The $\alpha$-divergences 
are $f$-divergences for the generator $f_\alpha(u)=\frac{4}{1-\alpha^2}(1-u^{\frac{1+\alpha}{2}})$:
\begin{eqnarray*}
I_\alpha(p:q) &=& I_{f_\alpha}(p:q)=\int p(x)\frac{4}{1-\alpha^2}\left( 1 - \left(\frac{q}{p}\right)^{\frac{1+\alpha}{2}}\right)\dmu(x),\\
&=& \frac{4}{1-\alpha^2} \left(1-\int p^{\frac{1-\alpha}{2}}(x) q^{\frac{1+\alpha}{2}}(x)\dmu(x) \right).
\end{eqnarray*}

The derivatives for the {\em $\alpha$-generators} are
$$
f_\alpha^{(i)}(u)=-\frac{2}{1-\alpha}\binom{\frac{1+\alpha}{2}}{i}  u^{\frac{1+\alpha}{2}-i}, \quad i\geq 2,
$$
where
$$
\binom{\gamma}{i} \eqdef \left\{
\begin{array}{ll}
\frac{\gamma (\gamma-1) \times \dots (\gamma-i+1)}{i!} & \mbox{if $i\leq \gamma$},\\
0 & \mbox{otherwise}.
\end{array}
\right.
$$
are the {\em generalized binomial coefficients} for any $\gamma\in\bbR$.
Symmetric divergence $I_0$ is known as the squared Hellinger distance:
$$
I_0(p:q) = 4 \left(1-\int \sqrt{p(x)} \sqrt{q(x)}\dmu(x) \right)
 =  2\int (\sqrt{p(x)}-\sqrt{q(x)})^2 \dmu(x)=I_0(q:p) .
$$

Therefore it follows that the power chi expansions for the $\alpha$-divergences is:
$$
I_{\alpha,k}^\chi(p:q) =  \sum_{i=2}^k  -\frac{2}{1-\alpha} \frac{1}{i!} \binom{\frac{1+\alpha}{2}}{i} \chi^\pm_i(p:q).
$$

Consider $\frac{1+\alpha}{2}$ an integer (say, $\alpha=2k-1$ for $k\in\{2, 3, \ldots\}$).
We can directly apply Theorem~\ref{thm:polygen} for the polynomial generator $f_{2k-1}(u)=\frac{1}{4k(1-k)} (1-u^k)$.

In that case, the iterated derivatives $f_{2k-1}^{(i)}$ are zero when $i>2k-1$.
Indeed, we have 
\begin{eqnarray*}
f_{2k-1}(u) &=&\frac{1}{4k(1-k)} (1-u^k)\\
f_{2k-1}^{(1)}(u) &=& \frac{u^{k-1}}{4(k-1)}\\
f_{2k-1}^{(2)}(u) &=& \frac{u^{k-2}}{4}\\
\vdots &=& \vdots\\
f_{2k-1}^{(i)}(u)&=&  \frac{(k-2)!}{(k-i)!} \frac{u^{k-i}}{4}\\
\vdots &=& \vdots\\
f_{2k-1}^{(k+1)}(u)&=&0.
\end{eqnarray*}
Hence, we have $f_{2k-1}^{(i)}(u)=0$ for $i>k>0$.
Notice that $f_{2k-1}^{(2)}(u)$ is positive on $u>0$, hence $f_{2k-1}$ is strictly convex.

Thus the $(2k-1)$-divergences have finite chi expansions:
\begin{eqnarray}
I_{2k-1}(p:q) &=& \sum_{i=2}^{k+1}  \frac{(k-2)!}{i! (k-i)!} \frac{1}{4} \chi^\pm_i(p:q),\\
&=& \frac{1}{k(1-k)} \left(1-\int p^{1-k}(x) q^{k}(x) \dmu(x) \right).
\end{eqnarray}

The $\alpha$-divergences between members $p=f(x;\theta_p)$ and $q=f(x;\theta_q)$ of the same affine exponential family 
admit the following closed-form formula for any $\alpha\in\bbR\backslash\{\pm 1\}$:

\begin{equation}
\impformula{
I_\alpha(p:q)= \frac{4}{1-\alpha^2} \left(1-\exp\left({F\left(\frac{1-\alpha}{2}\theta_p+\frac{1+\alpha}{2}\theta_q\right)-\left(\frac{1-\alpha}{2}F(\theta_p)+\frac{1+\alpha}{2}F(\theta_q)\right)}\right) \right).
}
\end{equation}

This formula matches the formula obtained by the chi expansion when $\alpha=2k-1$:
\begin{eqnarray*}
I_{2k-1}(p:q) &=& \sum_{i=2}^{k+1}  \frac{(k-2)!}{i! (k-i)!} \frac{1}{4}  \sum_{j=0}^i {(-1)}^{i-j} \binom{i}{j}   \exp\left({F\left((1-j)\theta_p+j\theta_q\right)-\left((1-j)F(\theta_p)+jF(\theta_q)\right)}\right),\\
&=&   \frac{1}{4 k (k-1)}   \sum_{i=2}^{k+1} \binom{k}{i} \sum_{j=0}^i {(-1)}^{i-j} \binom{i}{j}   \exp\left({F\left((1-j)\theta_p+j\theta_q\right)-\left((1-j)F(\theta_p)+jF(\theta_q)\right)}\right),
 \end{eqnarray*}
 since $\frac{(k-2)!}{i! (k-i)!}= \frac{1}{k (k-1)}  \binom{k}{i}$.

\section{Analytic formula and power chi series}

A function (or formula) is said analytic if it is locally defined by a convergent power series.
For example, the exponential function $\exp(x)=\sum_{i=0}^\infty \frac{x^i}{i!}$ 
or the function
$\sin(x) = \sum_{i=0}^\infty \frac{(-1)^i x^{2i+1}}{(2i + 1)!}$
are (globally) analytic on $\bbR$.
We have $\log (1+x)= \sum_{i=1}^\infty (-1)^{i+1} \frac{1}{i} x^i$ for $|x|<1$ (locally analytic at center $x=0$ with radius $1$).
That is, when $|x|\geq 1$, the power series $\sum_{i=1}^\infty (-1)^{i+1} \frac{1}{i} x^i$ diverge.
The series $\sum_{i=1}^\infty (-1)^{i+1} \frac{1}{i} x^i$ is a {\em power series} for the function $\log (1+x)$.
A Taylor expansion yields in the limit a {\em Taylor series} when the Taylor polynomials converge and the Taylor remainder tends to zero.
Log-normalizers of exponential families are analytic functions. 
Similarly, a chi expansion tends to a chi series when the remainder tends to zero and the weighted sum of power chi pseudo-distances converges.
Let us mention that there are cases where $I_f(p:q)$ may admit a power chi series but not $I_f(q:p)$ (or the reverse divergence $I_f^r(p:q)=I_f(q:p)$).

The JS divergence between two Gaussian distributions cannot be approximated by a finite  chi expansion because the sum of
weighted $i$-th power chi divergences diverges:
$$
\JS(m_1,m_2)\not= \sum_{i=2}^{\infty} (-1)^{i-2} \frac{1}{i(i-1)} \left(1-\frac{1}{2^{i-1}}\right) 
\sum_{j=0}^i (-1)^{i-j} \binom{i}{j}   \exp\left(\frac{j(j-1)}{2}\|m_1-m_2\|^2\right).
$$

This fact is in accordance with~\cite{Watanabe-2004} which proved that the KL divergence between Gaussian Mixture Models (GMMs) is {\em not} analytic. See Appendix~A of~\cite{GMMLSE-2016} for an English translation of the non-analytic KL GMM Japanese proof of~\cite{Watanabe-2004}.

Let us state a {\em selected} part of the Theorem~1~of~\cite{FdivTaylor-2002}:

\begin{theorem}[\cite{FdivTaylor-2002}]\label{thm:boundedratio}
Let $f:(0,\infty)\rightarrow \bbR$ be $k$ times differentiable and such that $f^{(k)}$ is absolutely continuous on $[m,M]$,
where $0<m\leq 1\leq M<\infty$, and assume that $m\leq \frac{q(x)}{p(x)}\leq M$ almost surely on $\calX$, then we have
$$
\left|I_f(p:q)-\sum_{i=2}^k  \frac{f^{(i)}(1)}{i!} \chi_i^\pm(p:q) \right| \leq \frac{1}{(k+1)!} \|f^{(k+1)}\|_{\infty}
\chi_{k+1}(p:q),
$$
where 
$$
\chi_{s}(p:q) = \int \frac{|q(x)-p(x)|^{s-1}}{p^s(x)} \dmu(x),\quad s\geq 2
$$
and $\|f^{(k+1)}\|_{\infty}=\ess \sup_{u\in [m,M]} |f^{(k+1)}(u)|$ is the Lebesgue $L_\infty$-norm (and $f^{(k+1)}\in L_\infty[m,M]$).
Furthermore, we have $\chi_{k+1}(p:q) \leq (M-m)^{k+1}$.
\end{theorem}

Observe that for bounded ratio densities, we have
$$
(-1)^i (1-m)^i \leq \chi^\pm_i(p:q) \leq (M-1)^i,
$$
and
$$
\chi_i(p:q) \leq \min((M-1)^i,(1-m)^i).
$$

Let us get a better upper bound:
We have $\frac{|q(x)-p(x)|^i}{p(x)^i}= |\frac{q}{p}-1| \frac{|q(x)-p(x)|^{i-1}}{p(x)^{i-1}}$, and 
$m-1\leq \frac{q}{p}-1\leq M-1$.
Therefore $|\frac{q}{p}-1| \leq \min(M-1,1-m)\leq M-m$.
It follows that $\chi_i(p:q) \leq (M-m)\chi_{i-1}(p:q)$ and $\chi_2(p:q)\leq (M-m) \int |q(x)-p(x)| \dmu(x)\leq (M-m)^2$.
Therefore $\chi_i(p:q)\leq (M-m)^i$.
Thus when $M<m+1$, we have $\lim_{i\leftarrow\infty} \chi_i(p:q) =0$.

To ensure that the Taylor remainder of the power chi series converge to zero, we need to have 
$$
\lim_{k\rightarrow\infty} \frac{1}{(k+1)!} \|f^{(k+1)}\|_{\infty} (M-m)^{k+1} =0.
$$
This is a sufficient (but not necessary) condition.

Notice that using Stirling approximation formula we have $(k+1)! \simeq \sqrt{2\pi(k+1)} \exp((k+1)\log \frac{(k+1)}{e})$  for large values of $k$.

This theorem does not apply to Gaussian distributions because we cannot bound the density ratio of two Gaussians.
However, we can {\em truncate} the Gaussians distributions on a compact support $\calD\subset\calX$, 
and define the extrema of the density ratio of two Gaussians as follows:
\begin{eqnarray*}
m(\theta_1:\theta_2) &=& \inf_{x\in\calD} \frac{p_\calD(x;\theta_2)}{p_\calD(x;\theta_1)},\\
M(\theta_1:\theta_2) &=& \sup_{x\in\calD} \frac{p_\calD(x;\theta_2)}{p_\calD(x;\theta_1)},
\end{eqnarray*}
where $p_\calD$ is the truncated distribution~\cite{TEF-2017}:
$$
p_\calD(x;\theta)=\frac{p(x;\theta)}{W_\calD(\theta)} = \frac{\exp(t(x)^\top\theta + k(x))}{\int_\calD \exp(t(x)^\top\theta + k(x))}\ 1_{\calD}(x),
$$
where $W_\calD(\theta)$ denotes the probability mass inside the domain $\calD$: $W_\calD(\theta)=\int_\calD p(x;\theta)\dmu(x)$.
The truncation of a regular (i.e., topologically open natural parameter space) and steep exponential family (i.e. mean parameter space 
coinciding with the interior of the closed convex hull of the support of the distributions, like the exponential family of Gaussian distributions) may not yield a regular and steep exponential family, see~\cite{truncGaussian-1994}.

Sufficient conditions to apply Theorem~\ref{thm:boundedratio} in order to get an analytic formula are to check that
\begin{enumerate}
\item  the remainder tends to zero when $k\rightarrow 0$, and
\item the power chi series converges when $k\rightarrow 0$.
\end{enumerate}

Notice that these two requirements are always fulfilled for finite discrete distributions for bounded $|f^{(i)}|$'s.


\subsection{Vanilla case study: The Bernoulli distributions}

Let $p=(\lambda_p,1-\lambda_p)$ and $q=(\lambda_q,1-\lambda_q)$ be two Bernoulli distributions (i.e., categorical distributions with two choices) with parameters $\lambda_p,\lambda_q\in (0,1)$. 
The $f$-divergence between $p$ and $q$ is
$$
I_f(p:q) = \lambda_p f\left(\frac{\lambda_q}{\lambda_p}\right) +  (1-\lambda_p) f\left(\frac{1-\lambda_q}{1-\lambda_p}\right).
$$

Let $m(\lambda_p:\lambda_q)=\min(\frac{\lambda_q}{\lambda_p},\frac{1-\lambda_q}{1-\lambda_p})$ 
and $M(\lambda_p:\lambda_q)=\max(\frac{\lambda_q}{\lambda_p},\frac{1-\lambda_q}{1-\lambda_p})$.

We have 
$$
\chi^\pm_{i}(\lambda_p:\lambda_q)=
 \frac{(\lambda_q-\lambda_p)^{i-1}}{\lambda_p^{i-1}} + \frac{(\lambda_p-\lambda_q)^{i-1}}{(1-\lambda_p)^{i-1}}
$$
That is, we do not use the binomial expansion for computing the chi pseudo-distances.

We consider the $f$-divergence for the exponential generator $f_{\mathrm{exp}}(u)=e^x-ex$ with
the $k$-th power chi expansion:  
$I_{f_{\mathrm{exp},k}}^\chi(\lambda_p:\lambda_q)  = \sum_{i=2}^k  \frac{e}{i!}   \chi_i^\pm(\lambda_p:\lambda_q)$.

When $k\rightarrow\infty$, the remainder  tends to zero, and the power chi expansions converge to $I_{f_{\mathrm{exp}}}(\lambda_p:\lambda_q)$:
$I_{f_{\mathrm{exp}}}^\chi(\lambda_p:\lambda_q)  = \sum_{i=2}^\infty  \frac{e}{i!}   \chi_i^\pm(\lambda_p:\lambda_q)$.

Let $\lambda_p=\eps\leq\frac{1}{2}$ and $\lambda_q=1-\eps$.
Then $m=\frac{\eps}{1-\eps}$ and $M=\frac{1-\eps}{1-\eps}$ 
so that $M-m=\frac{1-2\eps}{\eps(1-\eps)}$.

We have
$$
\chi^\pm_i(p:q)=  \frac{(1-2\eps)^{i}}{\eps^{i-1}} + \frac{(2\eps-1)^{i}}{(1-\eps)^{i-1}}.
$$

Since we have $\|f^{(k+1)}\|_\infty=\exp(\frac{1-\eps}{1-\eps})$, it follows that the remainder is upper bounded by
$$
\frac{\exp(\frac{1-\eps}{1-\eps})}{(k+1)!} \left(\frac{1-2\eps}{\eps(1-\eps)}\right)^{k+1}.
$$

In practice, we face numerical precision for computing the factorials $i!$ and large power of numbers close to zero.
In our experiments, we compute the $k$-th power chi expansions for $k=30$.

In our Java implementation,  the run for $\lambda_p=0.9$ and $\lambda_q=0.3$ yields
the exponential divergence $E(\lambda_p:\lambda_q)=108.20108519696437$ (closed-form formula), and  we list the experimental results obtained for the power chi pseudo-distances and approximations by power chi series in Table~\ref{tab:resexp}.

\begin{table}
{\small
\begin{tabular}{ll}
$
\begin{array}{ll}
i &  \chi_i(\lambda_p=0.9,\lambda_q=0.3) \\ \hline
2 & 4.000000000000002 \\
3 & 21.333333333333357 \\
4 & 129.77777777777794 \\
5 & 777.4814814814827 \\
6 & 4665.6790123456885 \\
7 & 27993.547325102943 \\
8 & 167961.6351165985 \\
9 & 1007769.5765889378 \\
10 & 6046617.615607398 \\
11 & 3.627970558959522\  10^7 \\
12 & 2.1767823360693753\  10^8 \\
13 & 1.3060694015953817\  10^9 \\
14 & 7.836416409603122\  10^9 \\
15 & 4.70184984575982 \ 10^{10} \\
16 & 2.821109907456029\  10^{11} \\
17 & 1.6926659444736094\  10^{12} \\
18 & 1.0155995666841666\  10^{13} \\
19 & 6.093597400105001\  10^{13} \\
20 & 3.6561584400630025\  10^{14} \\
21 & 2.1936950640378022\  10^{15} \\
22 & 1.3162170384226816\  10^{16} \\
23 & 7.8973022305360928\  10^{16} \\
24 & 4.7383813383216576\  10^{17} \\
25 & 2.8430288029929964\  10^{18} \\
26 & 1.7058172817957982\  10^{19} \\
27 & 1.0234903690774793\  10^{20} \\
28 & 6.140942214464877\  10^{20} \\
29 & 3.684565328678928\  10^{21} \\
30 & 2.2107391972073577\ 10^{22} 
\end{array} 
$
&
$
\begin{array}{lll}
i & E_i^\chi(\lambda_p=0.9,\lambda_q=0.3) & |R_i^\chi\chi(\lambda_p=0.9,\lambda_q=0.3)| \\ \hline
2 & 5.436563656918093 & 102.76452154004627 \\
3 & 15.101565713661374 & 93.099519483303 \\
4 & 29.800423008291787 & 78.40066218867258 \\
5 & 47.412204533912885 & 60.788880663051486 \\
6 & 65.02696908486016 & 43.174116112104215 \\
7 & 80.1250546023077 & 28.07603059465667 \\
8 & 91.44864241519754 & 16.75244278176683 \\
9 & 98.9976992034349 & 9.203385993529466 \\
10 & 103.52713339328994 & 4.673951803674427 \\
11 & 105.99773385339799 & 2.203351343566382 \\
12 & 107.23303408384565 & 0.9680511131187188 \\
13 & 107.8031726517244 & 0.39791254523997566 \\
14 & 108.04751775224481 & 0.15356744471955608 \\
15 & 108.14525579245294 & 0.05582940451142804 \\
16 & 108.18190755753099 & 0.019177639433380023 \\
17 & 108.19484347461736 & 0.006241722347013479 \\
18 & 108.19915544697947 & 0.0019297499848960342 \\
19 & 108.20051712246224 & 5.68074502126592\  10^{-4} \\
20 & 108.20092562510708 & 1.595718572957594\  10^{-4} \\
21 & 108.20104234014846 & 4.2856815909431134\  10^{-5} \\
22 & 108.20107417152339 & 1.102544098330327\  10^{-5} \\
23 & 108.20108247536032 & 2.721604047906112\  10^{-6} \\
24 & 108.20108455131955 & 6.456448176095364\  10^{-7} \\
25 & 108.20108504954977 & 1.4741459608558216\  10^{-7} \\
26 & 108.20108516452598 & 3.243839330480114\  10^{-8} \\
27 & 108.20108519007624 & 6.888129178150848\  10^{-9} \\
28 & 108.2010851955513 & 1.4130705494608264\  10^{-9} \\
29 & 108.20108519668408 & 2.802948984026443\  10^{-10} \\
30 & 108.20108519691063 & 5.374545253289398\  10^{-11}
\end{array}
$
\end{tabular}
}
\caption{Experimental results of the power chi expansions for the exponential $f$-divergence between Bernoulli distributions.
\label{tab:resexp}}
\end{table}

We report our implementation in {\sc Maxima}\footnote{\url{http://maxima.sourceforge.net/}} which can use high-precision arithmetic 
in Appendix~\ref{appendix:expf}.

For the JS generator, the series may either converge (e.g., $\lambda_p=0.1$ and $\lambda_q=0.05$) or diverge 
(e.g.  $\lambda_p=0.05$ and $\lambda_q=0.85$).

This vanilla case study can be extended to the case of multinoulli distributions (e.g., discrete distributions with $S$ choices or multinomials with one trial).

Computing $N$ $f$-divergences (for generators $f_1, \ldots, f_n$) between two discrete distributions with $S$ choices require
 $O(NS)$ operations. By precomputing a base of $K$ pseudo-distances  $\chi_2^\pm, \ldots, \chi_{k-1}^\pm$, we can approximate those $N$ $f$-divergences in $O(KN)$-time. This  yields a fast approximation scheme for batch approximations of $f$-divergences when $K\ll S$.

\subsection{Case study: Poisson distributions}

The probability mass function (pmf) of a Poisson distribution with parameter $\lambda$ is 
$$
f(x;\lambda) = \frac{\lambda^x e^{-\lambda}}{x!},
$$ 
for $x\in \bbN_0=\bbN\cup\{0\}=\{0,1,2,\ldots\}$.
The Poisson distributions form a {\em Discrete Exponential Family}~\cite{MLE-DEF-1971} (DEF) with 
natural parameter $\theta=\log\lambda$ (with $\theta\in\mathbb{R}$) and log-normalizer $F(\theta)=e^\theta$.
That is, we can rewrite the pmf as 
$$
f(x;\theta) = \exp(x\theta-e^{\theta}-\log x!),
$$ 
with $\theta=\log \lambda$ the natural parameter~\cite{EF-2009} belonging to $\bbR$.
Thus the Poisson family is a discrete affine exponential family.
The power chi pseudo-distance between two Poissons distributions~\cite{ChiFDiv-2014} of parameters $\lambda_1$ and $\lambda_2$ is given by
$$
\chi_k^\pm(\lambda_1:\lambda_2) = \sum_{j=0}^k (-1)^{k-j} \binom{k}{j}  e^{\lambda_1  \left(\frac{\lambda_2}{\lambda_1}\right)^{j} + (j-1)\lambda_1 -  j\lambda_2}.
$$

The ratio density between two Poisson distributions is
$$
\left(\frac{\lambda_q}{\lambda_p}\right)^x \exp(\lambda_p-\lambda_q) = \exp(x(\theta_q-\theta_p)-e^{\theta_q}+e^{\theta_p}).
$$ 

Consider $\lambda_q<\lambda_p$ so that $\frac{\lambda_q}{\lambda_p}<1$ (or equivalently $\theta_q<\theta_p$).
The ratio is then maximized at $x=0$, and we have $M(\lambda_p:\lambda_q)=\exp(\lambda_p-\lambda_q)$.
Notice that $1<M(\lambda_p:\lambda_q)<\infty$.
The ratio is lower bounded by $0$ (i.e., $m(\lambda_p:\lambda_q)=0$).

Thus the exponential $f$-divergence $E(p:q)$ can be expressed by a series when $\lambda_q<\lambda_p$:

$$
E_k^\chi(\lambda_p:\lambda_q)= \sum_{i=2}^k  \frac{e}{i!}   \chi_i^\pm(\lambda_p:\lambda_q),\quad \lambda_q<\lambda_p,
$$

and the remainder is $R_{E,k}^\chi(p:q)\leq \frac{1}{(k+1)!}  \|f^{(k+1)}(u)\|_{\infty}  M(\lambda_p:\lambda_q)^{k+1}$.
Here, the problem is that $\|f^{(k+1)}(u)\|_{\infty}$ is not bounded since $f^{(k+1)}(u)=\exp(u)$.

However, if we truncate the support to $[a, a+1, \dots, b-1, b]$, we can bound $\|f^{(k+1)}(u)\|_{\infty}$ and calculate the $\chi_i^\pm$ pseudo-distances because we deal with the case of finite distributions.

\subsection{Case study: The truncated exponential distributions}
Consider the double-sided truncated exponential distributions:
The density of an exponential distribution is $p(x;\lambda)=\lambda e^{-\lambda x}$ for $\lambda>0$ on the support $(0,\infty)$.
Its cumulative distribution function is $\CDF(x;\lambda)=1-e^{-\lambda x}$.
The exponential distributions form an exponential family~\cite{EF-2009} with $t(x)=-x$, $\theta=\lambda$ and $F(\theta)=-\log \theta$ (a convex log-normalizer).
The mass function is $W_\calD(\theta)=e^{-a\theta}-e^{-b\theta}$ for an interval domain $\calD=[a,b]$.
The double-sided truncated exponential distributions has density:
$$
\exp\left({-\theta x-\log \frac{e^{-a\theta}-e^{-b\theta}}{\theta}}\right).
$$

When $b=\infty$, we get singly truncated exponential distributions:
We have $W_\calD(\theta) = e^{-a\theta}$ and 
$$
p_\calD(x;\theta) = \theta\exp( -\theta x+a\theta).
$$
This is an exponential family with log-normalizer $F_\calD(\theta)=a\theta-\log\theta$.

Symbolic computations yield closed-form solutions for the $\chi^\pm_i$ divergences (see Appendix).
For example, we have when $\theta_2>\frac{2}{3}\theta_1$:

$$
\chi^\pm_3(p_\calD(x;\theta_1):p_\calD(x;\theta_2))=
\frac{2 {{\theta_2}^{4}}-10 \theta_1\, {{\theta_2}^{3}}+18 {{\theta_1}^{2}}\, {{\theta_2}^{2}}-14 {{\theta_1}^{3}}\, \theta_2+4 {{\theta_1}^{4}}}{{{\theta_1}^{2}}\, \left( 6 {{\theta_2}^{2}}-7 \theta_1\, \theta_2+2 {{\theta_1}^{2}}\right) }.
$$

\section{Conclusion and discussion}


\begin{sidewaystable}
\centering
{
\begin{tabular}{l||l|ll}
$I_f$ divergence name & generator $f$ & $k$-order power chi expansion\\ \hline\hline
generic $I_f$ & generic convex generator $f$ &  $I_{f,k}^\chi(p:q) = \sum_{i=2}^k \frac{f^{(i)}(1)}{i!} \chi_i^\pm(p:q)$\\ \hline
polynomial & $f_P(u)=\sum_{j=0}^d a_ju^j$ & $I_{f_P,k}^\chi(p:q) = \sum_{i=2}^d \left(\sum_{j=i}^d  a_j  \binom{j}{i}\right) \chi^\pm_i(p:q)$ \\
$\KL$ &  $f_\KL(u)=-\log u$ & $\KL_k^\chi(p:q) = \sum_{i=2}^k \frac{(-1)^i}{i} \chi_i^\pm(p:q)$\\ 
reverse $\KL$ & $f_{\KL^r}(u)=u\log u$ &  ${\KL^r}_k^\chi(p:q) = \sum_{i=2}^k  {(-1)}^{i} \frac{1}{i(i-1)} \chi_i^\pm(p:q)$\\	
$\alpha$-divergence & $f_\alpha(u)=\frac{4}{1-\alpha^2}(1-u^{\frac{1+\alpha}{2}})$ & $I_{\alpha,k}^\chi(p:q) =  \sum_{i=2}^k  -\frac{2}{1-\alpha} \frac{1}{i!} \binom{\frac{1+\alpha}{2}}{i} \chi^\pm_i(p:q)$\\
Jeffreys &   $f_J(u)=(u-1)\log u$  & $J_k^\chi(p:q) = \sum_{i=2}^k (-1)^{i}  \frac{1}{i-1}  \chi_i^\pm(p:q)$\\
Jensen-Shannon & $f_\JS(u)=-(u+1)\log \frac{1+u}{2} + u\log u$ & $\JS_k^\chi(p:q) = \sum_{i=2}^k (-1)^{i-2} \frac{1}{i(i-1)} \left(1-\frac{1}{2^{i-1}}\right)  \chi_i^\pm(p:q)$\\
Harmonic & $f_H(u)=\frac{2u}{u+1}$ & $H_k^\chi(p:q) =  1 + \sum_{i=2}^k (-1)^{i+1} \frac{1}{2^{i}}   \chi_i^\pm(p:q)$\\
Exponential & $f_{\mathrm{exp}}(u)=e^x-x$ & $E_k^\chi(p:q)= \sum_{i=2}^k  \frac{e}{i!}   \chi_i^\pm(p:q)$\\ \hline
\end{tabular}
}

\caption{Examples of finite power chi expansions for common divergences.
\label{tab:chinomial}}

\end{sidewaystable}

On one hand, it has been proven that the Kullback-Leibler divergence between Gaussians mixture models 
is in general not analytic~\cite{Watanabe-2004}.
On the other hand, we can express the $f$-divergences using power chi expansions~\cite{ChiFDiv-2014} yielding to power chi series when these expansions converge.
When the $f$-divergence generator is a polynomial, we obtain a closed-form formula for the $f$-divergence between members of the same exponential families with {\em affine} natural parameter space (e.g. isotropic Gaussian, Poisson, von Mises-Fisher, etc.) by using the finite power chi expansions.
This polynomial case includes the $\alpha$-divergences for $\alpha=2k-1$ where $k$ is an integer greater than $2$.

Table~\ref{tab:chinomial} summarizes the power chi expansions of some common $f$-divergences.
Observe that when we precompute or approximate the first $k$ power chi pseudo-distances, 
we can calculate or approximate quickly  the $k$-order power chi expansions of $f$-divergences using 
the basis of $\chi^\pm_i$ pseudo-distances (see Table~\ref{tab:chinomial}).
Thus, we can efficiently  approximate a batch  of $f$-divergences using $\chi^\pm_i$ look-up tables.

Truncating distributions~\cite{truncGaussian-1989,truncGaussian-1994} on a finite compact domain let us bound both the ratio of densities and the chi pseudo-distances,  potentially yielding approximation formul\ae{} for the $f$-divergences. 
Another direction is to approximate the $f$-generator by a polynomial~\cite{ApproxTheory-2003} and then to apply the finite chi expansion of the $f$-divergence induced by that polynomial approximation of the generator.

Finally, to conclude, let us mention that it is interesting to find examples where the  entropy of mixture distributions is analytic (besides the case of mutinomial distributions) since this would also yield analytic formula for the Jensen-Shannon divergence according to Eq.~\ref{eq:jsh}.

\vskip 1cm
\noindent Additional materials are available online at:\\
\centerline{\url{https://franknielsen.github.io/PowerChiExpansionsFdiv/}}

\appendix

\section{Information geometry and $f$-divergences\label{sec:stdfdiv}}

Since $I_{\beta f}=\beta I_{f}$ for $\beta>0$, we fix the scale of the $f$-divergence by choosing $\beta$ that $\beta f''(1)=1$.
Indeed, theory and applications usually consider relative distances instead of absolute distances.
Thus we select the representative $f$-divergence in the family $\{f_{\alpha,\beta}(u)=\beta (f(u)+\alpha(u-1))\}$ of generators (with $f(1)=0$) such that $f'(1)=0$ and $f''(1)=1$. Such a $f$-divergence is called a standard $f$-divergence in Amari's textbook~\cite{IG-2016} (p. 56):
 The standard $f$-divergences are invariant divergences (\cite{IG-2016}, p. 54) which satisfies the property of information monotonicity by coarse-graining. 
At infinitesimal scale, a first-order Taylor expansion of the standard $f$-divergence between distributions belonging to a parametric family $\{p(x;\theta)\}$ yields half of the squared Mahalanobis distance\footnote{For a positive definite matrix $Q\succ 0$, the  Mahalanobis  metric (distance) is $M_Q(p,q)=\sqrt{(p-q)^\top Q (p-q)}$.} for the Fisher information matrix $G_\theta$
(\cite{IG-2016}, p. 62):
$$
I_f(p(x;\theta):p(x;\theta+\dtheta))= \frac{1}{2} M_{G_\theta}^2(\dtheta,\dtheta) = \frac{1}{2}  \dtheta^\top G_\theta \dtheta = \frac{1}{2}  \sum_{i,j}  g_{ij}(\theta) \dtheta_i \dtheta_j,
$$
where
$$
G_\theta = [g_{ij}(\theta)],\quad g_{ij}(\theta)=E_{p(x;\theta)}\left[\frac{\partial}{\partial\theta_i}\log p(x;\theta)\frac{\partial}{\partial\theta_j}\log p(x;\theta)\right].
$$

The $f$-divergence can be extended to positive measures\footnote{For standard $f$-divergence, we let $I_f(p^+:q^+)=I_f(\bar{p}:\bar{q})+q^+-p^+$,
 where $\bar{p}=\frac{p^+}{\int p(x)\dmu(x)}$ and $\bar{q}=\frac{q^+}{\int q(x)\dmu(x)}$ are normalized probability densities.} instead of probability measures~\cite{DivEstimator-2006}.

The $f$-divergence is a separable divergence that can be written as: 
$$
I_f(p:q) = \int_\calX i_f(p(x):q(x))  \dmu(x),
$$
with $i_f(a:b)= a f\left(\frac{b}{a}\right)$ for $a,b>0$.
The function $i_f$ is called the perspective function~\cite{GeneralizedPerspectiveI-2005} of the convex function $f$.

\section{Some sanity checks using a computer algebra system}

We use the Computer Algebra System (CAS) {\sc Maxima}.\footnote{\url{http://maxima.sourceforge.net/}}
Some code has also been written in Python using Sympy~\cite{sympy}.

\subsection{Checking the closed-form formula for $\frac{f^{(k)}(1)}{k!}$}

\begin{verbatim}
/* k-th order derivative of the Jensen-Shannon f-generator
		evaluated at one and divided by k! */ 
k:23;
fJSderivative(i) := ((-1)**(i-2))*(1-1/(2**(i-1)))*(1/((i*(i-1))));
fJSderivative(23);

fJS(u) := -(u+1)*log((1+u)/2) + u*log(u);
at(diff(fJS(u),u,k),u=1)/k!;
\end{verbatim}

We find: $f_\JS^{(23)}(1)=-\frac{182361}{92274688}$.

\subsection{Chi pseudo-distances between truncated exponential distributions}

\begin{verbatim}
theta1 : 1;
theta2 : 2;
a : 2;
b : 3;
/* truncated exponential distribution */
p(x,theta) := theta*exp(-theta*x);
mass(theta) := exp(-a*theta)-exp(-b*theta);
q(x,theta) := theta*exp(-theta*x)/(exp(-a*theta)-exp(-b*theta));
k:10;
integrate(  (q(x,theta2)-q(x,theta1))**k/(q(x,theta1))**(k-1),x,a,b);
\end{verbatim}

\subsection{Power chi expansions for the exponential $f$-divergence between binomials\label{appendix:expf}}

\begin{verbatim}
/* Exponential f-divergence */
f(u) := exp(u)-(u*%e);
coeff(i) := %e/i!;

/* f-divergence between binomials */
fdiv(p,q) := (p*f(q/p)) + ((1-p)*f((1-q)/(1-p)));

/* i-order chi pseudo-distance */
chipm(i,p,q) :=  (((q-p)**i)/(p**(i-1))) + (((p-q)**i)/((1-p)**(i-1))) ;

/* power chi series */
fapprox(p,q,max) := sum(coeff(i)*chipm(i,p,q), i, 2,max);

/* test */ 
p: 0.2; q: 0.99;
result: float(fdiv(p,q));

for i: 2 step 1 thru 30 do 
bloc(
approx: float(fapprox(p,q,i)),
err: float(abs(result-approx)),
display(i,approx,result,err)
);
\end{verbatim}



\begin{thebibliography}{10}

\bibitem{TEF-2017}
Masafumi Akahira.
\newblock {\em Statistical Estimation for Truncated Exponential Families}.
\newblock Springer, 2017.

\bibitem{IG-2016}
Shun-ichi Amari.
\newblock {\em Information geometry and its applications}, volume 194.
\newblock Springer, 2016.

\bibitem{IGDiv-2010}
Shun-ichi Amari and Andrzej Cichocki.
\newblock Information geometry of divergence functions.
\newblock {\em Bulletin of the Polish Academy of Sciences: Technical Sciences},
  58(1):183--195, 2010.

\bibitem{FdivTaylor-2005}
George~A Anastassiou.
\newblock Higher order optimal approximation of {C}sisz\'ar's $f$-divergence.
\newblock {\em Nonlinear Analysis: Theory, Methods \& Applications},
  61(3):309--339, 2005.

\bibitem{FdivTaylor-2002}
Neil~S Barnett, Pietro Cerone, Sever~Silvestru Dragomir, and Anthony Sofo.
\newblock Approximating csisz\'ar $f$-divergence by the use of {T}aylor's
  formula with integral remainder.
\newblock {\em Mathematical Inequalities and Applications}, 5:417--434, 2002.

\bibitem{SI-2006}
Manuel Bronstein.
\newblock {\em Symbolic integration I: transcendental functions}, volume~1.
\newblock Springer Science \& Business Media, 2006.

\bibitem{measure-2013}
Marek Capinski and Peter~E Kopp.
\newblock {\em Measure, integral and probability}.
\newblock Springer Science \& Business Media, 2013.

\bibitem{Csiszar-1967}
Imre Csisz{\'a}r.
\newblock Information-type measures of difference of probability distributions
  and indirect observation.
\newblock {\em studia scientiarum Mathematicarum Hungarica}, 2:229--318, 1967.

\bibitem{truncGaussian-1994}
Joan Del~Castillo.
\newblock The singly truncated normal distribution: a non-steep exponential
  family.
\newblock {\em Annals of the Institute of Statistical Mathematics},
  46(1):57--66, 1994.

\bibitem{GAN-2014}
Ian Goodfellow, Jean Pouget-Abadie, Mehdi Mirza, Bing Xu, David Warde-Farley,
  Sherjil Ozair, Aaron Courville, and Yoshua Bengio.
\newblock Generative adversarial nets.
\newblock In {\em Advances in neural information processing systems}, pages
  2672--2680, 2014.

\bibitem{truncGaussian-1989}
Laxman~M Hegde and Ram~C Dahiya.
\newblock Estimation of the parameters in a truncated normal distribution.
\newblock {\em Communications in Statistics-Theory and Methods},
  18(11):4177--4195, 1989.

\bibitem{EFvonMises-2013}
Kurt Hornik and Bettina Gr{\"u}n.
\newblock On conjugate families and {J}effreys priors for von {M}ises--{F}isher
  distributions.
\newblock {\em Journal of statistical planning and inference}, 143(5):992--999,
  2013.

\bibitem{Liese-2006}
Friedrich Liese and Igor Vajda.
\newblock On divergences and informations in statistics and information theory.
\newblock {\em IEEE Transactions on Information Theory}, 52(10):4394--4412,
  2006.

\bibitem{JS-1991}
Jianhua Lin.
\newblock Divergence measures based on the {S}hannon entropy.
\newblock {\em IEEE Transactions on Information theory}, 37(1):145--151, 1991.

\bibitem{GeneralizedPerspectiveI-2005}
Pierre Mar{\'e}chal.
\newblock On a functional operation generating convex functions, {P}art 1:
  {D}uality.
\newblock {\em Journal of Optimization Theory and Applications},
  126(1):175--189, 2005.

\bibitem{sympy}
Aaron Meurer, Christopher~P Smith, Mateusz Paprocki, Ond{\v{r}}ej
  {\v{C}}ert{\'\i}k, Sergey~B Kirpichev, Matthew Rocklin, AMiT Kumar, Sergiu
  Ivanov, Jason~K Moore, Sartaj Singh, et~al.
\newblock {SymPy}: symbolic computing in {P}ython.
\newblock {\em PeerJ Computer Science}, 3:e103, 2017.

\bibitem{StatMinkowski-2019}
Frank Nielsen.
\newblock The statistical {M}inkowski distances: Closed-form formula for
  {G}aussian mixture models.
\newblock {\em arXiv preprint arXiv:1901.03732}, 2019.

\bibitem{EF-2009}
Frank Nielsen and Vincent Garcia.
\newblock Statistical exponential families: A digest with flash cards.
\newblock {\em arXiv preprint arXiv:0911.4863}, 2009.

\bibitem{ECE-EF-2010}
Frank Nielsen and Richard Nock.
\newblock Entropies and cross-entropies of exponential families.
\newblock In {\em 2010 IEEE International Conference on Image Processing},
  pages 3621--3624. IEEE, 2010.

\bibitem{ChiFDiv-2014}
Frank Nielsen and Richard Nock.
\newblock On the chi square and higher-order chi distances for approximating
  $f$-divergences.
\newblock {\em IEEE Signal Processing Letters}, 21(1):10--13, 2014.

\bibitem{GMMLSE-2016}
Frank Nielsen and Ke~Sun.
\newblock Guaranteed bounds on information-theoretic measures of univariate
  mixtures using piecewise log-sum-exp inequalities.
\newblock {\em Entropy}, 18(12):442, 2016.

\bibitem{DivEstimator-2006}
Tomoaki Nishimura and Fumiyasu Komaki.
\newblock The information geometric structure of generalized empirical
  likelihood estimators.
\newblock {\em Communications in Statistics—Theory and Methods},
  37(12):1867--1879, 2008.

\bibitem{fgan-2016}
Sebastian Nowozin, Botond Cseke, and Ryota Tomioka.
\newblock $f$-{GAN}: Training generative neural samplers using variational
  divergence minimization.
\newblock In {\em Advances in neural information processing systems}, pages
  271--279, 2016.

\bibitem{chisquared-1900}
Karl Pearson.
\newblock On the criterion that a given system of deviations from the probable
  in the case of a correlated system of variables is such that it can be
  reasonably supposed to have arisen from random sampling.
\newblock {\em The London, Edinburgh, and Dublin Philosophical Magazine and
  Journal of Science}, 50(302):157--175, 1900.

\bibitem{ApproxTheory-2003}
George~M Phillips.
\newblock {\em Interpolation and approximation by polynomials}, volume~14.
\newblock Springer Science \& Business Media, 2003.

\bibitem{MLE-DEF-1971}
Steve Selvin.
\newblock Maximum likelihood estimation in the truncated, single parameter,
  discrete exponential family.
\newblock {\em The American Statistician}, 25(1):41--42, 1971.

\bibitem{diBruno-2005}
Karlheinz Spindler.
\newblock A short proof of the formula of {F}a{\`a} di {B}runo.
\newblock {\em Elemente der Mathematik}, 60(1):33--35, 2005.

\bibitem{Watanabe-2004}
Sumio Watanabe, Keisuke Yamazaki, and Miki Aoyagi.
\newblock Kullback information of normal mixture is not an analytic function.
\newblock {\em IEICE technical report. Neurocomputing}, 104(225):41--46, 2004.

\end{thebibliography}
\end{document}